\documentclass[envcountsect,envcountsame,a4paper,UKenglish,11pt]{article}

\usepackage{amsmath,amssymb, amsfonts, amsbsy, pstricks,bbm,mathrsfs,multirow,bigstrut,fixmath,enumerate,graphicx,bm }
\usepackage{amsthm}
\usepackage{microtype}

\newcommand{\indexlabeling}{\eta}
\newcommand{\afunction}{f}

\newcommand{\tensorlabeling}{\mathbf{g}}

\newcommand{\classformulas}{\mathcal{F}}

\newcommand{\variableset}{X}
\newcommand{\Yvariableset}{Y}

\newcommand{\C}{{\mathbb{C}}}

\newcommand{\R}{\mathbb{R}}
\newcommand{\carvingdecomposition}{\mathcal{T}}
\newcommand{\carvingwidth}{\mathit{carw}}

\newcommand{\treewidth}{\mathbold{tw}}

\newtheorem{theorem}{Theorem}[section]

\newtheorem{proposition}[theorem]{Proposition}
\newtheorem{lemma}[theorem]{Lemma}

\newtheorem{observation}[theorem]{Observation}

\newtheorem{definition}[theorem]{Definition}

\newcommand{\N}{{\mathbb{N}}}

\newcommand{\edfunction}{\delta}

\newcommand{\nodes}{\mathit{nodes}}
\newcommand{\arcs}{\mathit{arcs}}
\newcommand{\ket}[1]{|#1\rangle}
\newcommand{\bra}[1]{\langle #1 |}

\newcommand{\graph}{G}

\newcommand{\finalvalue}{\mathcal{V}}

\newcommand{\inputvertices}{\mathit{In}}
\newcommand{\outputvertices}{\mathit{Out}}
\newcommand{\internalvertices}{\mathit{Mid}}
\newcommand{\atensor}{g}

\newcommand{\indexset}{\mathcal{I}}
\newcommand{\rank}{\mathit{rank}}

\newcommand{\leaves}{\mathit{leaves}}
\newcommand{\width}{\mathbf{w}}

\newcommand{\vertexlabelingfunction}{{\theta}}
\newcommand{\edgelabelingfunction}{{\xi}}

\newcommand{\degree}{{\mathrm{deg}}}
\newcommand{\totaldegree}{{\mathrm{tdeg}}}

\newcommand{\tensorContraction}{\mathit{Contr}}

\newcommand{\tensornetwork}{\mathcal{N}}

\newcommand{\trace}{{\mathit{Tr}}}

\newcommand{\valuevertex}{\valuevertex}

\newcommand{\monomials}{\mathcal{M}}

\newcommand{\realpart}{a}
\newcommand{\complexpart}{b}
\newcommand{\realpartvar}{a}
\newcommand{\complexpartvar}{b}
\newcommand{\boolassignment}{\alpha}
\newcommand{\polynomial}{p}

\usepackage[margin=1in]{geometry}

\begin{document}

\title{A Near-Quadratic Lower Bound for the Size of Quantum Circuits of Constant Treewidth}

\author{Mateus de Oliveira Oliveira \\ University of Bergen \\ mateus.oliveira@uib.no}

\maketitle

\begin{abstract}
We show that any quantum circuit of treewidth $t$, built from $r$-qubit gates, requires 
at least $\Omega(\frac{n^{2}}{2^{O(r\cdot t)}\cdot \log^4 n})$ gates to compute the 
element distinctness function. Our result generalizes a near-quadratic lower bound for quantum formula size 
obtained by Roychowdhury and Vatan [SIAM J. on Computing, 2001]. 
The proof of our lower bound follows by an extension of Ne\v{c}iporuk's method  to the context of 
quantum circuits of constant treewidth. This extension is made via a combination of 
techniques from structural graph theory, tensor-network theory, and the connected-component counting 
method, which is a classic tool in algebraic geometry. 
\\
\\
{{\bf Keywords: } Super-Linear Lower Bounds, Quantum Circuits, Algebraic Tensor Networks, \\
\hphantom{blablablablab} Treewidth}
\end{abstract}

\section{Introduction}

Proving superlinear lower bounds on the size of circuits computing some function in NP 
remains one of the greatest challenges of computational complexity theory 
\cite{FindGolovnevHirschKulikov2015,IwamaMorizumi2002,Lachish2001Explicit}. Currently, the best known lower 
bound for a function in NP  is of the order of $5n-o(1)$ for Boolean circuits with gates from the binary De-Morgan basis 
\cite{IwamaMorizumi2002,Lachish2001Explicit} and of the order of $(3+1/86)n+o(n)$ for Boolean circuits 
with arbitrary fan-in-2 gates \cite{FindGolovnevHirschKulikov2015}. 
Therefore, research in this direction has focused on lower bounds for restricted classes of circuits. 
In particular, superlinear lower bounds have been proved for Boolean formulas, and for formulas
constructed from non-Boolean gates. 
The strongest known size lower bound for Boolean formulas over the complete binary basis, which is 
of the order of $\Omega(n^2/\log n)$, is due to Ne\v{c}iporuk~\cite{Neciporuk1966} and remains unimproved for four decades.
If we restrict ourselves to formulas over the De Morgan basis ($\wedge$,$\vee$,$\neg$), then 
the best known lower bound is of the order of $n^{3-o(1)}$ \cite{Hastad1998}. 
Tur\'{a}n and Vatan proved an $\Omega(n^2/\log^2 n)$ size lower bound for arithmetic formulas, 
and an $\Omega(n^{3/2}/\log n)$ size lower bound for threshold formulas \cite{TuranVatan1997computation}. 
Yao introduced the notion of quantum formulas 
(i.e. quantum circuits whose whose underlying graph is a tree) and proved a slightly superlinear
lower bound on the size of quantum formulas computing the majority function \cite{Yao1993}. 
Subsequently, Roychowdhury and Vatan proved an $\Omega(n^2/\log^2 n)$ size lower bound for quantum 
formulas \cite{RoychowdhuryVatan2001quantum}.

The treewidth of a graph is a parameter that has played a central role in several branches 
of algorithmics, combinatorics and structural graph theory 
\cite{RobertsonSeymour1984Treewidth,DemaineFominHajiaghayiThilikos2005,ArnborgLagergrenSeese1991,
ArnborgProskurowski1989,Courcelle1990Monadic}. The notion of treewidth has also caught attention from 
the circuit complexity community due to the fact that the satisfiability of read-once\footnote{A circuit or formula 
is read-once if each variable labels at most one input vertex.} Boolean circuits
of constant treewidth can be determined in polynomial time \cite{AlekhnovichRazborov2002,AllenderChenLouPeriklisPapakonstantinouTang2014,
BroeringLokamSatyanarayana2004,GalJing-Tang2012,GeorgiouKonstantinosPapakonstantinou2008,HeLiangSarma2010,JansenSarma2010}.
Recently, near-quadratic lower bounds were shown for Boolean circuits of constant treewidth \cite{deOliveiraOliveira2015Satisfiability}. 
In the context of quantum computation, it has been shown that the satisfiability of read-once quantum circuits of constant 
treewidth can be determined in polynomial time \cite{deOliveiraOliveira2015Satisfiability}. Additionally, in a 
pioneering result, Markov and Shi have shown that quantum circuits of constant treewidth can be simulated with 
multiplicative precision in polynomial time \cite{MarkovShi2008}.

In this work we prove near-quadratic size lower bounds for quantum circuits of constant treewidth. 
More precisely, our main result (Theorem \ref{theorem:MainTheoremQuantumCircuits}) states that 
any quantum circuit of treewidth $t$, built from $r$-qubit gates, 
requires at least $\Omega(\frac{n^2}{2^{O(r\cdot t)}\cdot \log^{4} n})$ gates to compute the 
$n$-bit element distinctness function. 
In particular, our result imply near-quadratic size lower bounds for several natural restrictions 
of circuits. For instance, formulas have treewidth at most $1$, TTSP 
series-parallel\footnote{Another notion of series-parallel circuits studied in circuit complexity 
theory is the notion of Valiant series parallel circuits, for which no superlinear lower bounds are 
known \cite{Valiant1977,Calabro2008}.} circuits have treewidth at most $2$, and $k$-outerplanar 
circuits have treewidth $O(k)$. Additionally, our result implies superlinear lower bounds even for circuits of treewidth 
$c\cdot \log n$ for some sufficiently small constant $c$.
Our lower bound can be regarded as a simultaneous generalization of  
superlinear lower bounds provided in \cite{RoychowdhuryVatan2001quantum} for the size of quantum formulas 
and in $\cite{deOliveiraOliveira2016Tradeoffs}$ for the size of Boolean circuits of constant treewidth. 

It is worth noting that our results do not follow from previous super-linear lower bounds. 
Although it has been shown that quantum formulas of size $S$ can be simulated 
by Boolean {\em circuits} of size $S^{O(1)}$ \cite{RoychowdhuryVatan2001quantum}, it is a 
long-standing open problem to determine whether quantum formulas can be polynomially simulated 
by Boolean {\em formulas} of size $S^{O(1)}$.  Such an efficient simulation result has 
been been obtained only for read-once quantum formulas \cite{ConsentinoKothariPaetznick2013}. Nevertheless, 
the techniques in \cite{ConsentinoKothariPaetznick2013} fail if the read-once condition is removed. 
Similarly, it has been shown in \cite{MarkovShi2008} that quantum circuits of treewidth $t$ and size $S$ 
can be simulated by Boolean circuits of size $2^{O(t)}\cdot S^{O(1)}$. Nevertheless the Boolean circuits 
obtained by the simulation in \cite{MarkovShi2008} have unbounded treewidth due to the fact that
this simulation uses multiplication of large numbers. Indeed, it is an open problem 
to determine whether quantum circuits of treewidth $t$ can be 
polynomially simulated by Boolean circuits of treewidth $f(t)$ for some function $f:\N\rightarrow \N$.
Therefore, our superlinear lower bounds for quantum circuits of constant treewidth do 
not follow from superlinear lower bounds for Boolean circuits of constant treewidth 
obtained in \cite{deOliveiraOliveira2016Tradeoffs}. Additionally, it is not known either 
whether quantum (resp. Boolean) circuits of treewidth $t$ can be polynomially simulated 
by quantum (resp. Boolean) circuits of treewidth $t-1$. In particular, it is not known whether 
quantum circuits of treewidth $t$ can be polynomially simulated by quantum formulas. Therefore,
our results are not implied by the superlinear lower bounds for quantum formulas obtained 
in \cite{RoychowdhuryVatan2001quantum}.

\section{Proof Techniques}
\label{section:ProofTechniques}

To prove our lower bound, we will extend Ne\v{c}iporuk's method to 
the context of quantum circuits of constant treewidth. This method, which was 
originally devised by Ne\v{c}iporuk to prove superlinear size lower bounds 
for Boolean formulas \cite{Neciporuk1966}, has been generalized to several models of 
computation, including arithmetic and threshold formulas \cite{TuranVatan1997computation}, 
quantum formulas \cite{RoychowdhuryVatan2001quantum} and Boolean circuits of 
constant treewidth \cite{deOliveiraOliveira2016Tradeoffs}. However, to extend Ne\v{c}iporuk's 
method to the context of quantum circuits of constant treewidth, we will need to 
introduce new tools which combine techniques from structural graph theory, 
tensor network theory, and algebraic geometry.

The challenging part in generalizing Ne\v{c}iporuk's method to a class of formulas $\classformulas$
is a step which has been termed {\em path squeezing} in \cite{RoychowdhuryVatan2001quantum}. Intuitively 
this step is used to show that if a function $f:\{0,1\}^{Y}\rightarrow \{0,1\}$ can be 
computed by a formula $F\in \classformulas$ which has at most $l$ leaves labeled with variables 
in $Y$, then $f$ can also be computed by a formula in $\classformulas$ of size at most $l^{O(1)}$. 
While this step can be solved easily on Boolean formulas, path squeezing becomes highly non-trivial 
on formulas with non-boolean gates, such as arithmetic and threshold formulas \cite{TuranVatan1997computation} 
and quantum formulas \cite{RoychowdhuryVatan2001quantum}. The interest in path squeezing stems from 
the fact that it allows us to establish an upper bound for the number of functions computable by 
formulas with at most $l$ input nodes labeled with variables. 

The path squeezing technique is intrinsic to formulas and does not generalize to Boolean circuits nor 
to Quantum circuits of treewidth $t>1$. This drawback was circumvented in \cite{deOliveiraOliveira2016Tradeoffs} for Boolean circuits
of constant treewidth. Although it is not known whether Boolean circuits of treewidth $t$ with $l$ inputs labeled 
by variables can be squeezed into a Boolean circuits of treewidth $t$ and size $l^{O(1)}$, 
it was shown in \cite{deOliveiraOliveira2016Tradeoffs} that 
each such circuit $C$ can always be compactly represented by a constraint satisfaction problem (CSP) with 
$O(t\cdot l)$ constant-width constraints representing the same function as $C$. This is enough to establish 
an upper bound on the number of Boolean functions which can be computed by circuits of constant 
treewidth with at most $l$ input vertices labeled with variables. Unfortunately, the mapping
from circuits to CSPs does not generalize to the context of quantum circuits. 

To provide an analog squeezing technique for quantum circuits of constant treewidth, we will
generalize the notion of tensor network, which is widespread in quantum physics \cite{MarkovShi2008,Orus2014practical}, to 
the notion of algebraic tensor network. We will show that if a Boolean function $f:\{0,1\}^Y\rightarrow \{0,1\}$
can be computed by a Quantum circuit of treewidth at most $t$ with at most $l$ inputs labeled by variables
in $Y$, then such function $f$ can also be represented by an algebraic tensor network of rank $O(t)$ and 
size $O(t\cdot l)$. This step requires the development of a new contraction technique for tensor networks
that may be of independent interest. In order to upper bound the number of functions that can be represented by algebraic 
tensor networks of such size and rank, we will employ the connected component counting method, a classic 
tool in algebraic geometry introduced by Warren \cite{Warren1968lower}.

\section{Preliminaries}

We assume familiarity with basic concepts of quantum computation (see for instance \cite{NielsenChuang2010}).
For completeness, we briefly define the notion of quantum circuit. A {\em qubit} 
is a unit vector in $\C^{2}$. We let $\{\ket{0},\ket{1}\}$ be the standard orthonormal basis of $\C^2$.
A $k$-qubit {\em quantum gate} is a unitary matrix $U\in \C^{2^{k}\times 2^{k}}$. 
A $1$-qubit measurement element is a matrix $M\in \C^{2\times 2}$ such that both $M$ and $I-M$ are positive semidefinite. 
A quantum circuit over a set of variables $X$ is a directed acyclic graph (DAG)
$C = (V,E,\vertexlabelingfunction,\edgelabelingfunction)$, where 
$V$ is a set of vertices, $E$ is a set of edges, 
$\vertexlabelingfunction$ is a function that labels vertices in $V$ with quantum gates, with 
variables in $X$ or with some element in $\{\ket{0},\ket{1}\}$, and $\edgelabelingfunction:E\rightarrow \{1,...,|E|\}$
is a bijection that labels edges in $E$ with numbers in $\{1,...,E\}$. 
The vertex set is partitioned into a set of input vertices $\inputvertices$, a set of 
internal vertices $\internalvertices$, and a set of output vertices $\outputvertices$. 
A quantum circuit is subject to the following constraints.  

\begin{enumerate}
	\item If $v$ is an input vertex, then $v$ has in-degree $0$ and out-degree $1$. Additionally, 
		$\vertexlabelingfunction(v)\in X\cup \{\ket{0},\ket{1}\}$. 
	\item If $v$ is an internal vertex, then for some $k$, $v$ has $k$ in-neighbours and $k$-out neighbours. 
		additionally, $\vertexlabelingfunction(v)$ is a unitary gate acting on $k$ qubits.
	\item If $v$ is an output vertex, then $v$ has in-degree $1$ and out-degree $0$. Additionally, 
		 $\vertexlabelingfunction(v)$ is a $1$-qubit measurement element. 
\end{enumerate}

We note that a quantum circuit may have multiple edges with same source vertex and target vertex. We also note that a variable 
$x\in \variableset$ may label several input nodes of $C$ (Fig. \ref{figure:quantumCircuit}).

\begin{figure}[hf]
\centering
\includegraphics[scale=0.50]{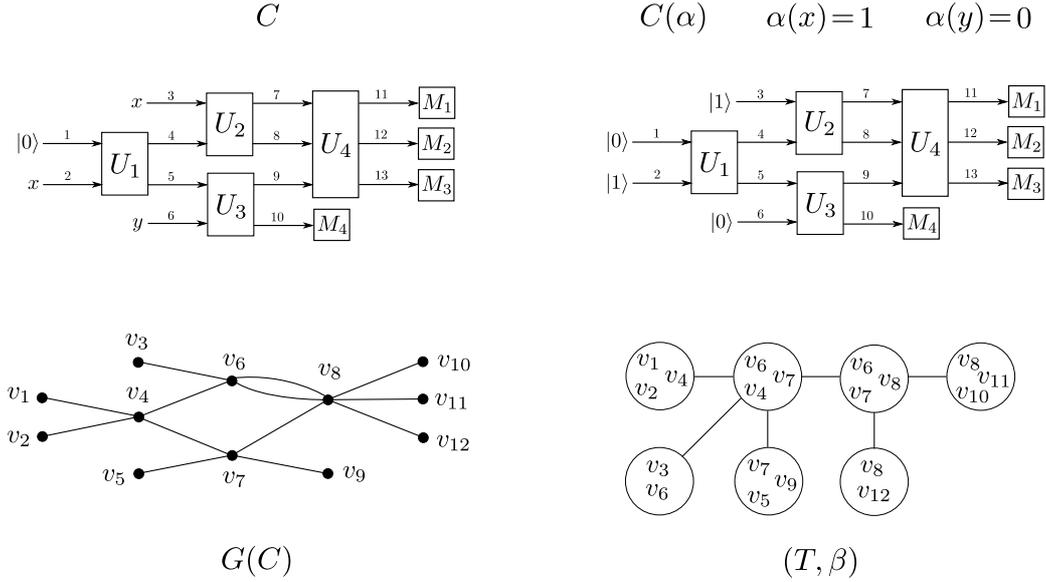}
\caption{A quantum circuit $C$ over a set of variables $\variableset =\{x,y\}$. The quantum circuit $C(\alpha)$ is 
obtained by initializing the inputs of $C$ according to the assignment $\alpha$ which sends $x$ to $1$ and $y$ to $0$. 
$\graph(C)$ is the graph associated with associated with $C$. $(T,\beta)$ is a tree decomposition of $\graph(C)$ of 
width $2$.}
\label{figure:quantumCircuit}
\end{figure}

We will use quantum circuits as a model of computation for Boolean functions. 
A Boolean assignment for a set of variables $\variableset$ is a function $\alpha:X\rightarrow \{0,1\}$. 
We denote by $\{0,1\}^{\variableset}$ the set of all Boolean assignments for $\variableset$. 
A Boolean function over $\variableset$ is a function $f:\{0,1\}^X\rightarrow \{0,1\}$. 
If $C$ is a quantum circuit with $m$ input vertices, then the internal vertices of $C$ naturally define a unitary 
matrix $U_{C} \in \C^{2^{m}\times 2^{m}}$ and the output vertices of $C$ 
define a measurement element $M = \bigotimes_{u\in\outputvertices} \vertexlabelingfunction(u)$ in $\C^{2^{m}\times 2^{m}}$. 
Additionally, if all input nodes of $C$ are labeled with qubits in $\{\ket{0},\ket{1}\}$, then 
these input nodes define a basis state $\ket{\psi}=\bigotimes_{u\in \inputvertices}\vertexlabelingfunction(u)$ 
in $\C^{2^m}$. In this case, the output probability of $C$ is defined as 
$\mathit{Pr}(C) = \mathit{Tr}(U_C\ket{\psi}\bra{\psi}U_C^{\dagger},M_C)$. On the other hand, if some input 
nodes of $C$ are labeled with variables in $\variableset$, and $\alpha\in \{0,1\}^{\variableset}$ is a 
Boolean assignment for $\variableset$, then we let $C(\alpha)$ be the 
quantum circuit obtained by initializing each input vertex whose label is a variable $x\in \variableset$
with the basis state $\ket{\alpha(x)}$ (Fig. \ref{figure:quantumCircuit}). 
The output probability of $C$ on input $\alpha$ is defined as the 
output probability of the circuit $C(\alpha)$.

\begin{definition}[Function Computed by a Quantum Circuit]
\label{definition:FunctionFromCircuit}
We say that a quantum circuit $C$ over a set of variables $\variableset$ 
computes a Boolean function $f:\{0,1\}^{\variableset}\rightarrow \{0,1\}$ if the following conditions 
are satisfied for each assignment $\alpha\in \{0,1\}^{\variableset}$.
\begin{enumerate}
	\item If $f(\alpha) = 1$ then $\mathit{Pr}(C(\alpha))> 1/2$. 
	\item If $f(\alpha) = 0$ then $\mathit{Pr}(C(\alpha))< 1/2$. 
\end{enumerate}
\end{definition}

If $C$ is a quantum circuit, then we let $\graph(C)$ be the underlying undirected graph of $C$, which is 
obtained by forgetting edge directions as well as vertex and edge labels. We note that the multiplicities 
of edges of $C$ are preserved in $\graph(C)$ (Fig. \ref{figure:quantumCircuit}). 

\begin{definition}
\label{definition:TreeDecomposition}
Let $G = (V,E)$ be an undirected graph, possibly containing multiple edges. 
A tree decomposition of $G$ is a pair $(T,\beta)$ where $T$ is a tree, 
and $\beta:\nodes(T)\rightarrow 2^{V}$ satisfying the following properties.
\begin{itemize}
	\item $\bigcup_{u\in \nodes(T)} \beta(u) =  V$, 
	\item for every edge $\{v,v'\} \in E$, there is a node $u\in \nodes(T)$ such that $\{v,v'\} \subseteq \beta(u)$,
	\item for every vertex $v\in V$, the set $\{u\in \nodes(T) \;|\; v\in \beta(u)\}$ induces a connected 
		subtree of $T$. 
\end{itemize}
\end{definition}

The {\em width} of a tree decomposition $(T,\beta)$ is defined as  $\width(T,\beta)=\max_{u}\{|\beta(u)|-1\}$.
The {\em treewidth} of $G$, denoted by $\treewidth(G)$, is the minimum width of a tree decomposition of $G$.
The treewidth of a quantum circuit $C$ is defined as the treewidth of its underlying undirected graph $\graph(C)$ 
(Fig. \ref{figure:quantumCircuit}).

\section{Algebraic Tensors and Algebraic Tensor-Networks}
\label{section:TensorNetwork}

Tensors and tensor-networks have been used as a fundamental tool for the simulation of quantum 
systems and quantum circuits \cite{MarkovShi2008,Orus2014practical}. In this section we define 
the notions of {\em algebraic tensors} and {\em algebraic tensor networks}. While a tensor 
is a multidimensional array of complex numbers, an algebraic tensor is a multidimensional array of complex 
polynomials. An algebraic tensor network is a collection of algebraic tensors. We will use 
such networks as a model of computation for Boolean functions. 
If a function $f:\{0,1\}^{\variableset}\rightarrow \{0,1\}$ can be computed by a 
quantum circuit of size $S$ and treewidth $t$, then $f$ can also be computed by an 
algebraic tensor network of size $S$ and treewidth $t$.  
Therefore, superlinear size lower-bounds for algebraic tensor networks of treewidth $t$
imply superlinear size lower bounds for quantum circuits of treewidth $t$. 

Let $\Pi = \{\,\ket{0}\bra{0},\,\ket{0}\bra{1},\,\ket{1}\bra{0},\,\ket{1}\bra{1}\,\}$ be the 
standard orthonormal basis for the space of $2\times 2$ complex matrices.
Let $\variableset$ be a finite set of variables. We denote by $\C[\variableset]$ the ring of complex polynomials 
in $\variableset$, and by $\R[\variableset]$ the ring of real polynomials in $\variableset$.

\begin{definition}[Algebraic Tensor]
\label{definition:AlgebraicTensor}
An {\em algebraic tensor} with index set $\indexset=\{i_1,...,i_k\}$ over a finite set of variables 
$\variableset$ is a $k$-dimensional array
$\atensor = [\atensor(\sigma_{i_1},...,\sigma_{i_k})]_{\sigma_{i_1},...,\sigma_{i_k}}$ where for 
each $\sigma_{i_1}...\sigma_{i_k}\in \Pi^k$, 
the entry $\atensor(\sigma_{i_1},...,\sigma_{i_k})$ is a polynomial in $\C[\variableset]$. 
\end{definition}

We note that $\atensor$ has $4^k$ entries. We write $\indexset(\atensor)$ to denote the index set of 
$\atensor$. The rank of $\atensor$ is defined as $\rank(\atensor) = |\indexset(\atensor)|$, i.e., as the 
size of the index set of $\atensor$. As a degenerate case, we regard a polynomial 
$p\in \C[\variableset]$ as an algebraic tensor of rank $0$. In other words, a polynomial is an 
algebraic tensor with empty index set. The {\em algebraic degree} of $\atensor$, denoted by $\degree(\atensor)$, is defined as 
the maximum degree of a polynomial occurring in $\atensor$. 

\begin{definition}[Algebraic Tensor Network]
\label{definition:AlgebraicTensorNetwork}
An {\em algebraic tensor network} over $\variableset$ is a sequence ${\tensornetwork = [g_1,g_2,...,g_m]}$ of algebraic tensors 
over $\variableset$ such that 
${|\{j\;|\; i\in \indexset(g_j)\}| = 2}$ for each ${i\in \bigcup_{j=1}^{m} \indexset(g_j)}$. 
\end{definition}

In other words, if a number $i$ occurs in the index set of some tensor in $\tensornetwork$, then $i$ occurs in the 
index set of precisely two such tensors.
The size of $\tensornetwork$, denoted by $|\tensornetwork|$, is defined as the number of tensors in $\tensornetwork$. 
The rank of $\tensornetwork$ is defined as $\rank(\tensornetwork) = \max_{i} \rank(\atensor_i)$. 
The {\em algebraic degree} of $\tensornetwork$ is defined as 
$\degree(\tensornetwork) = \max_{i}\degree(\atensor_i)$, and the {\em total degree} of $\tensornetwork$ is 
defined as $\totaldegree(\tensornetwork) = \sum_{i} \degree(g_i)$.

\begin{figure}[hf]
\centering
\includegraphics[scale=0.23]{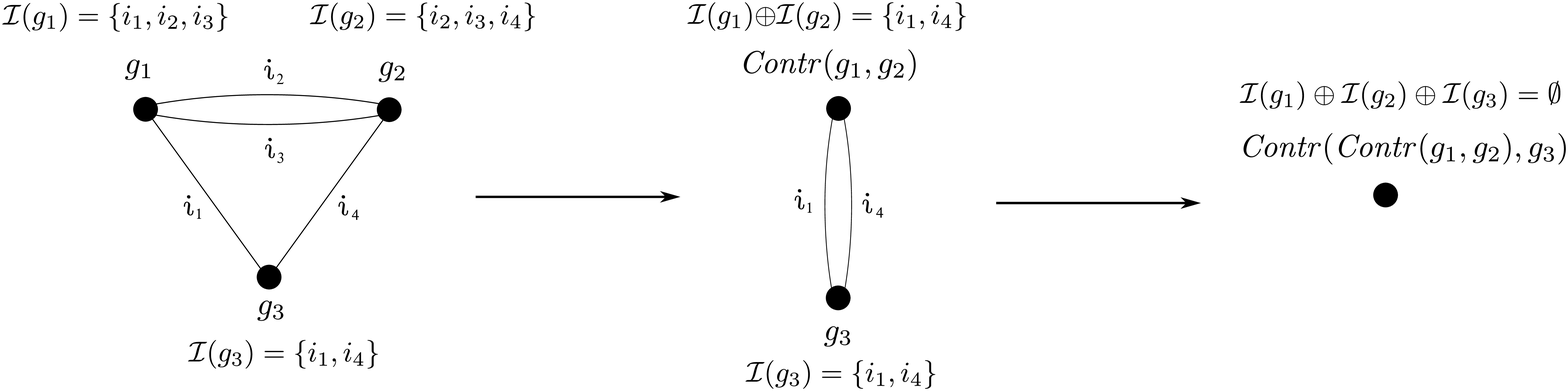}
\caption{Left: the graph $\graph(\tensornetwork)$ of an algebraic tensor network 
$\tensornetwork =[\,\atensor_1,\atensor_2,\atensor_3 \,]$. 
Middle: contracting the tensors $\atensor_1$ and $\atensor_2$ yields the algebraic tensor network 
$\tensornetwork =[\, \atensor_3,\tensorContraction(\atensor_1,\atensor_2)\,]$. 
Right: after all pairs have been contracted, the only remaining algebraic tensor 
$\tensorContraction(\tensorContraction(\atensor_1,\atensor_2),\atensor_3)$
is a complex polynomial, i.e., an algebraic tensor of rank $0$ (its index set is empty).} 
\label{figure:AlgebraicTensorNetworkContraction}
\end{figure}

An algebraic tensor network $\tensornetwork = [\atensor_1,\atensor_2,...,\atensor_m]$ can be represented by a
labeled undirected graph 
$\graph(\tensornetwork) = (V,E,\tensorlabeling,\indexlabeling)$ with vertex set $V = \{v_1,...,v_m\}$
and edge-set $E = \{e_i\;|\; i\in \bigcup_{j}\indexset(g_j)\}$. 
Each vertex $v_j\in V$ is labeled by $\tensorlabeling$ with the tensor $\tensorlabeling(v_j) = g_j$. 
Each edge $e_i$ is labeled by $\indexlabeling$ with the label $\indexlabeling(e_i) = i$. Finally, each 
edge $e_i$ has endpoints $v_j$ and $v_{j'}$ if and only if $i\in \indexset(g_j)\cap \indexset(g_{j'})$ 
(see Fig. \ref{figure:AlgebraicTensorNetworkContraction}). We note 
that $\graph(\tensornetwork)$ may have multiple edges, but no loops. 
We say that a tensor network $\tensornetwork$ is connected if the graph $\graph(\tensornetwork)$ is 
connected. In this work we will only be concerned with connected tensor networks. 
The treewidth of an algebraic tensor network $\tensornetwork$ is defined as the treewidth 
of its graph $\graph(\tensornetwork)$.

\subsection{Algebraic Tensor Network Contraction}
\label{subsection:AlgebraicTensorNetworkContraction}

Let $\indexset$ and $\indexset'$ be sets of positive integers, and let 
$\indexset\oplus \indexset' = (\indexset\cup \indexset')\backslash (\indexset\cap \indexset')$ be the 
symmetric difference between $\indexset$ and $\indexset'$. We say that a pair of algebraic 
tensors $\atensor$ and $\atensor'$ is {\em contractible} if ${\indexset(\atensor)\cap \indexset(\atensor') \neq \emptyset}$.
If $\atensor,\atensor'$ is a contractible pair of algebraic tensors such that 
${\indexset(\atensor)=\{i_1,...,i_k,l_1,...,l_r\}}$ and ${\indexset'=\{j_1,...,j_{k'},l_1,...,l_r\}}$,
then the contraction of $\atensor$ with $\atensor'$ is an algebraic tensor 
$\tensorContraction(\atensor,\atensor')$ with index set 
$\indexset(\atensor)\oplus \indexset(\atensor') = \{i_1,...,i_k, j_1,...,j_{k'}\}$ where for each 
$\sigma_{i_1},...,\sigma_{i_k},\sigma_{j_1},...,\sigma_{j_{k'}} \in \Pi^{k+k'}$, the entry 
$\tensorContraction(\atensor,\atensor')(\sigma_{i_1},...,\sigma_{i_k},\sigma_{j_1},...,\sigma_{j_{k'}})$ is 
defined as

\begin{equation}
\label{equation:TensorContraction}
\sum_{\sigma_{l_1},...,\sigma_{l_r}\in \Pi^r} \atensor(\sigma_{i_1},...,\sigma_{i_k},\sigma_{l_1},...,\sigma_{l_{r}})\cdot 
\atensor'(\sigma_{j_1},...,\sigma_{j_{k'}}, \sigma_{l_1},...,\sigma_{l_r}),
\end{equation}

The following observation follows straightforwardly from Equation \ref{equation:TensorContraction}.

\begin{observation}
\label{observation:AdditiveDegree}
Let $\atensor$ and $\atensor'$ be a pair of contractible tensors. Then 
$$\degree(\tensorContraction(\atensor,\atensor')) \leq \degree(\atensor)+\degree(\atensor').$$
\end{observation}

\begin{definition}
\label{definition:AlgebraicTensorNetworkContraction}
Let $\tensornetwork= [g_1,...,g_{m}]$ be an algebraic tensor network and let $g_j$ and $g_{l}$ be a pair of contractible tensors
in $\tensornetwork$. We say that a tensor network $\tensornetwork'$ is obtained by the contraction of 
$g_{j}$ and $g_{l}$ if $\tensornetwork' = (\tensornetwork \backslash \{g_j,g_{l}\}) \cup \{\tensorContraction(g_j,g_{l})\}$. 
\end{definition}

\newcommand{\finalpolynomial}{p}

The contraction of the tensors $g_{j}$ and $g_{l}$ in $\tensornetwork$ may be visualized as an operation 
that merges the vertices $v_j$ and $v_{l}$ in the graph $G(\tensornetwork)$ associated with 
$\tensornetwork$ (Fig. \ref{figure:AlgebraicTensorNetworkContraction}). The new vertex arising 
from the merging of $v_{j}$ and $v_{l}$ is now labeled with $\tensorContraction(g_j,g_{l})$. 
We note that if $\tensornetwork$ is connected, then the resulting tensor network 
$\tensornetwork'$ is also connected. Therefore, a tensor network $\tensornetwork$ with $m$ 
tensors can be contracted $m-1$ times until a unique tensor $\atensor$ 
is left (Fig. \ref{figure:AlgebraicTensorNetworkContraction}). The remaining tensor 
$g$ is an algebraic tensor of degree $0$ (i.e, $\atensor$ is  a complex polynomial).

Let $\tensornetwork$ be an (algebraic) tensor network of size $m$. We say that a sequence 
$\tensornetwork_0\tensornetwork_1...\tensornetwork_{m-1}$ 
is a contraction sequence for $\tensornetwork$ if $\tensornetwork_0 = \tensornetwork$ and for each $i\in \{1,...,m-1\}$, 
the tensor network $\tensornetwork_{i}$ is obtained from $\tensornetwork_{i-1}$ by the contraction of some 
pair of tensors. The next observation states that the algebraic tensor which arises from the contraction of 
all (algebraic) tensors in $\tensornetwork$ does not depend on the order of contraction. 

\begin{observation}
\label{observation:TensorContraction}
Let $\tensornetwork$ be an algebraic tensor network of size $m$. Let $\tensornetwork_1\tensornetwork_2...\tensornetwork_m$ 
and $\tensornetwork_1'\tensornetwork_2'...\tensornetwork_m'$ be contraction sequences for $\tensornetwork$.
Let $\tensornetwork_m = [g]$ and $\tensornetwork_m' = [\atensor']$. Then $\atensor = \atensor'$. 
\end{observation}

We note that the proof of Observation \ref{observation:TensorContraction} is identical to 
the proof that contracting all tensors of a tensor network, in any given order, yields the same 
outcome (see for instance \cite{MarkovShi2008,Orus2014practical}).

We let $\atensor_{\tensornetwork}$ be the rank-0 algebraic tensor obtained by the contraction of 
all algebraic tensors in $\tensornetwork$.
By Observation \ref{observation:TensorContraction}, this tensor is well defined.  
Let $\atensor_{\tensornetwork} = p_1 + i\cdot p_2$ where $p_1,p_2\in \R[X]$. The value 
of $\tensornetwork$ is defined as $\finalvalue_{\tensornetwork} = \sqrt{p_1^2 + p_2^2}$.

\begin{proposition}
\label{proposition:PolynomialFromAlgebraicTensorNetwork}
Let $X$ be a set of variables, and let $\tensornetwork = [g_1,...,g_m]$ be a connected algebraic tensor network 
over $\variableset$. Then $\finalvalue_{\tensornetwork}^2$ is a real polynomial in $\R[\variableset]$ of degree at most 
$2\cdot \totaldegree(\tensornetwork)$. 
\end{proposition}
\begin{proof}
Let $\atensor_{\tensornetwork} = p_1+i\cdot p_2$ where $p_1$ and $p_2$ are polynomials in $\R[\variableset]$. 
Then $\finalvalue_{\tensornetwork}^2 = p_1^2 + p_2^2$ is clearly a polynomial in $\R[\variableset]$.
By Observation \ref{observation:AdditiveDegree}, for any pair of contractible 
algebraic tensors $\atensor$ and $\atensor'$, it holds that 
$\degree(\tensorContraction(\atensor,\atensor')) \leq \degree(\atensor) + \degree(\atensor')$. Therefore, 
$\degree(\atensor_{\tensornetwork}) \leq \sum_{j=1}^{m} \atensor_j = \totaldegree(\tensornetwork)$.
This implies that the degree of $\finalvalue_{\tensornetwork}^{2}$ is at most $2\cdot \totaldegree(\tensornetwork)$. 
\end{proof}

 Note that if $\tensornetwork$ is an algebraic tensor network 
over $\variableset$ and $\alpha \in \{0,1\}^{\variableset}$ is a 
Boolean assignment of $\variableset$, then $\finalvalue_{\tensornetwork}(\alpha)$
is a positive real number.

\begin{definition}[Function Computed by an Algebraic Tensor Network]
We say that an algebraic tensor network $\tensornetwork$ over a set of variables $\variableset$ 
computes a function $f:\{0,1\}^{\variableset}\rightarrow \{0,1\}$ if the following conditions are 
verified for each assignment $\alpha\in \{0,1\}^{\variableset}$. 
\begin{enumerate}
	\item If $f(\alpha) = 1$ then $\finalvalue_{\tensornetwork}(\alpha) > 1/2$. 
	\item If $f(\alpha) = 0$ then $\finalvalue_{\tensornetwork}(\alpha) < 1/2$. 
\end{enumerate}
\end{definition}

Any function $f:\{0,1\}^{X}\rightarrow \{0,1\}$ that can be computed by a 
quantum circuit $C$ of treewidth $t$ can also be computed by an algebraic tensor 
network $\tensornetwork_{C}$ of treewidth $t$ and algebraic-degree $1$. This statement 
is formalized in the following proposition.  

\begin{proposition}
\label{proposition:ConversionQuantumCircuitsAlgebraicTensorNetworks}
Let $C$ be a quantum circuit over a set of variables $X$ of treewidth $t$ such that 
all gates in $C$ act on at most $r$ qubits. Then there is an algebraic tensor 
network $\tensornetwork_{C}$ over $X$ of treewidth $t$, algebraic degree $1$, and rank at most $2r$,
such that $\finalvalue_{\tensornetwork_{C}}(\alpha) = \mathit{Pr}(C(\alpha))$ for every assignment 
$\alpha:X\rightarrow \{0,1\}$. 
\end{proposition}

The Proof of Proposition \ref{proposition:ConversionQuantumCircuitsAlgebraicTensorNetworks} is
analogous to the conversion from variable-less quantum circuits to tensor networks provided in 
in \cite{MarkovShi2008}. For completeness, we include the construction of the algebraic tensor network 
$\tensornetwork_{C}$ in Appendix \ref{appendix:ConversionQuantumCircuitsAlgebraicTensorNetworks}.

\subsection{Reducing the Size of Algebraic Tensor Networks}
\label{section:TensorNetworkReduction}

Let $\variableset$ be a set of variables and $Y\subseteq X$. 
We say that a polynomial $p\in \C[\variableset]$ {\em constrains} a 
variable $y\in Y$ if $y$ occurs in some non-zero term of $p$. 
We say that an algebraic tensor $\atensor$ over $X$ is a $Y$-tensor if 
some polynomial in $\atensor$ constrains some variable in $y\in Y$.
In this section we define the notion of carving width of a graph. It can 
be shown that the carving width of a graph is at most a constant times 
its treewidth. Subsequently, we show that if $\tensornetwork$ is an algebraic tensor network 
computing a Boolean function $f:\{0,1\}^{Y}\rightarrow \{0,1\}$, then 
this function can also be computed by an algebraic tensor network $\tensornetwork'$
of size at most $4l(w+1)$ and rank at most $2w$,  where $l$ is 
the number of $Y$ tensors in $\tensornetwork$ and $w$ is the {\em carving width}
of the graph $\graph(\tensornetwork)$.

Let $T$ be a tree. We denote by $\nodes(T)$ the set of nodes of $T$, by $\arcs(T)$ the 
set of arcs of $T$. We say that a node $u\in \nodes(T)$ is a {\em leaf} if $u$ has no
children. If $u$ is not a leaf, then $u$ is said to be an {\em internal node} of $T$. 
We denote by $\leaves(T)$ the set of leaves of $T$. 
For each node $u\in \nodes(T)$, we let $T[u]$ denote the subtree of $T$ rooted at $u$.

\begin{definition}[Rooted Carving Decomposition]
\label{definition:CarvingDecomposition}
Let $G=(V,E)$ be an undirected graph, possibly containing multiple edges. 
A {\em rooted carving decomposition} of $G$ is a pair $(T,\gamma)$ where 
$T$ is a rooted {\em binary} tree and $\gamma:\leaves(T)\rightarrow V$ is a bijection mapping each leaf 
$u\in \leaves(\carvingdecomposition)$ to a single vertex $\gamma(u)\in V$. 
\end{definition}

Observe that the internal nodes of the tree $T$ are unlabeled. 
Given a node $u\in \nodes(T)$, we let $V(u)=\gamma(\leaves(T[u])) = \{\gamma(v)\;|\;v\in \leaves(T[u])\}$
be the image of the leaves of $T[u]$ under $\gamma$. 
For a subset $V' \subseteq V$ we let $E(V')$ denote the set of edges in $G$ with one endpoint 
in $V'$ and another endpoint in $V\backslash V'$. The width of $\carvingdecomposition$, denoted by 
$\carvingwidth(\carvingdecomposition)$, is defined as $\max\{|E(V(u))| : u\in \nodes(T)\}$. 
The carving width of a graph $G$, denoted by $\carvingwidth(G)$, is defined as the 
minimum width of a carving decomposition of $G$. The following lemma establishes a 
relation between carving width and treewidth of a graph.

\begin{lemma}[\cite{RobertsonSeymour1995}]
\label{lemma:CarvingVsTreewidth}
Let $G$ be an undirected graph of treewidth $t$ and maximum degree $\Delta$. There exists a rooted carving decomposition 
$(T,\gamma)$ of $G$ of width $O(\Delta\cdot t)$. 
\end{lemma}

Let $\tensornetwork$ be a tensor network and $\graph(\tensornetwork)=(V,E,\tensorlabeling,\indexlabeling)$ 
be the graph associated with 
$\tensornetwork$. Let $(T,\gamma)$ be a carving decomposition of $\graph(\tensornetwork)$ of width $w$. 
For each node $u\in \nodes(T)$, we define the following set.
$$\leaves(T[u],\Yvariableset) =  \{u'\in \leaves(T[u])\;|\;\tensorlabeling(\gamma(u'))\mbox{ is a $\Yvariableset$-tensor}\}.$$ 

In words, $\leaves(T[u],\Yvariableset)$ is the set of leaves $u'$ of $T$ whose corresponding vertex $\gamma(u')$ 
in $\graph(\tensornetwork)$ is labeled by $\tensorlabeling$ with a $\Yvariableset$-tensor. 

\begin{definition}[$\Yvariableset$-node]
\label{definition:Ynode}
We say that a node $u\in \nodes(T)$ 
is a $\Yvariableset$-node if $u$ is either a leaf such that 
$\tensorlabeling(\gamma(u))$ is a $\Yvariableset$-tensor, or if $u$ is an internal node $u\in \nodes(T)$ such that 
$\leaves(T[u.l],Y)\neq \emptyset \mbox{ and } \leaves(T[u.r],Y)\neq \emptyset$.
\end{definition}

We let $\nodes(T,Y)$ denote the set of all $Y$-nodes of $T$. For instance, 
in Fig. \ref{figure:EliminatingTensors} we depict a carving decomposition of some 
algebraic tensor network. In this decomposition, the $\Yvariableset$-nodes are indicated in red. 
If $u$ is a $\Yvariableset$-node, then we say that a node $u'\neq u$ is the $\Yvariableset$-parent of $u$ if 
$u'$ is the ancestor of $u$ at minimal distance from $u$ with the property that $u'$ is itself a 
$\Yvariableset$-node. Alternatively, we may say that $u$ is a $Y$-child of $u'$. 
The following lemma states that the number of $Y$-nodes in a carving decomposition is 
proportional to the number of $Y$-leaves in it.

\begin{figure}[hf]
\centering
\includegraphics[scale=0.23]{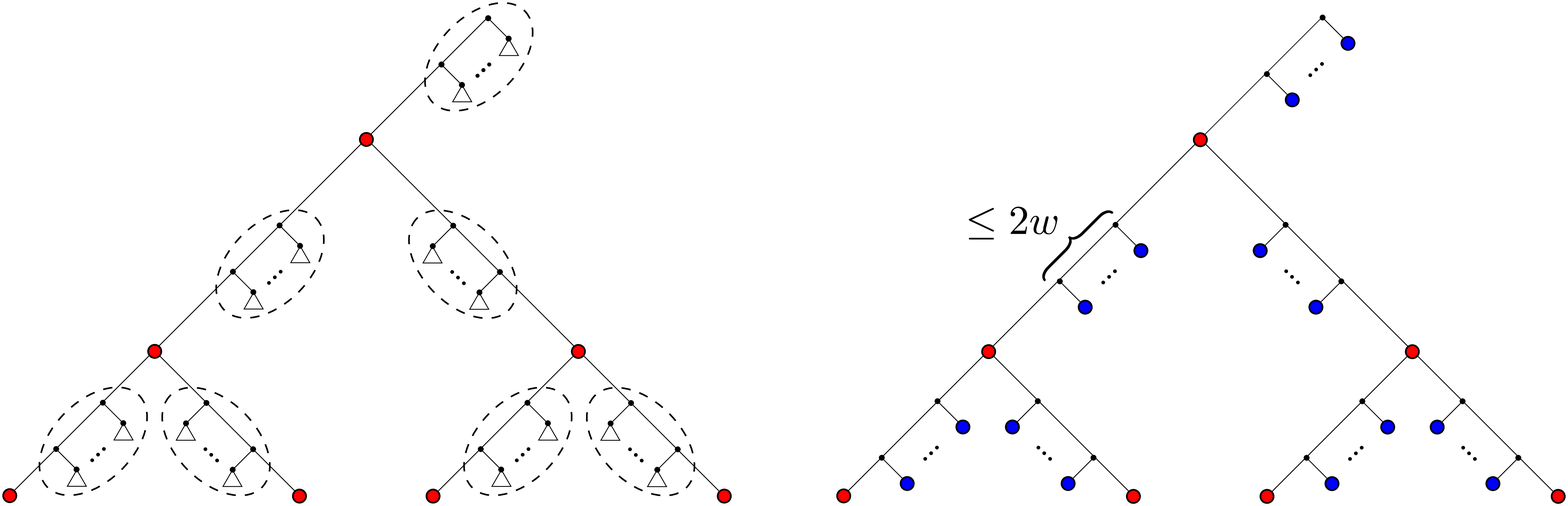}
\caption{
Left: A carving decomposition of the graph $\graph(\tensornetwork)$ associated 
with some tensor network $\tensornetwork$. The red nodes are $\Yvariableset$-nodes in $\nodes(T,Y)$.
The nodes inside each dashed region form a connected component  
$T_i$ of the forest $T\backslash \nodes(T,Y)$. Right: For each $T_i$, let $G[T_i]$ be 
the subgraph of $G(\tensornetwork)$ induced by the vertices $\gamma(\leaves(T_i))$. Then 
$G[T_i]$ has at most $2w$ connected components $C_{i,j}$. The contraction of all tensors 
labeling vertices of a component $C_{i,j}$ gives rise to a tensor $g_{i,j}$ of rank at most $2w$. 
Each such tensor corresponds to a blue node in the carving decomposition to the right. 
}
\label{figure:EliminatingTensors}
\end{figure}

\begin{lemma}
\label{lemma:TwoTimesMinusOneYNodes}
 $|\nodes(T,Y)|=2\cdot |\leaves(T,Y)|-1$.
\end{lemma}
\begin{proof}
Let $u$ be an internal $\Yvariableset$-node of $T$. We show that 
$u$ has precisely two $\Yvariableset$-children. Suppose for contradiction 
that $u$ has at most one $\Yvariableset$-child. Then by definition $u$ is not a $\Yvariableset$-node, since 
in this case either $\leaves(T[u.l],\Yvariableset)=\emptyset$ or $\leaves(T[u.r],\Yvariableset)=\emptyset$. Now suppose that 
$u$ has at least $3$ $\Yvariableset$-children. Since $T$ is a binary tree, two $\Yvariableset$-children of $u$ are either 
descendants of $u.l$ or descendants of $u.r$. Lets assume that $z$ and $z'$ are two distinct $\Yvariableset$-children of $u$
which are descendants of $u.l$. We observe that neither $z$ is a descendant of $z'$ nor $z'$ is a descendant of $z$, since 
otherwise, only one of these two vertices could have been a $\Yvariableset$-child of $u$. Now let $u'$ be the closest
ancestor of $z$ which is also an ancestor of $z'$. Then $u'$ is by definition a $\Yvariableset$-node. 
Since $u'$ is a descendant of $u.l$, this contradicts the assumption that $u$ is the $\Yvariableset$-parent of $z$ and $z'$.

Now let $T[\Yvariableset]$ be the tree whose nodes are $\Yvariableset$-nodes of $T$ and such that $(u,u')$ is an 
arc of $T[\Yvariableset]$ if and only 
if $u$ is the $\Yvariableset$-parent of $u'$. Then by the discussion above we have that $T[\Yvariableset]$ is a 
binary tree. Since any binary tree with $|\leaves(T,\Yvariableset)|$ leaves has $|\leaves(T,\Yvariableset)|-1$ 
internal nodes, the total number of $\Yvariableset$-nodes in $T$ is $2|\Yvariableset|-1$ (see Fig. \ref{figure:EliminatingTensors}). 
\end{proof}

Now let $T' = T \backslash \nodes(T,Y)$ be the forest which is obtained by deleting 
from $T$ all of its $\Yvariableset$-nodes. 

\begin{lemma}
\label{lemma:NumberComponentsTrees}
The number of connected components in the forest $T'= T\backslash \nodes(T,Y)$ 
is at most $|\nodes(T,Y)|=2|\leaves(T,Y)|-1$.
\end{lemma}
\begin{proof}
Let $T_1,...,T_k$ be the connected components of the forest $T' = T\backslash \nodes(T,Y)$. 
For each $i\in \{1,...,k\}$, let $r_i$ be the root of $T_i$, and let $u_i$ be the closest descendant of $r_i$ in 
$T$ which is a $Y$-node. 
We claim that $u_i$ is uniquely determined by $r_i$. To see this, assume for the sake 
of contradiction that there are two descendants $u_i$ and $u_i'$ of $r_i$ in $T$ with the property that 
$u_i$ and $u_i'$ are $Y$-nodes at a minimal distance from $r_i$. Let $u_i''$
be the closest ancestor of $u_i$ which is also an ancestor of $u_i'$. Then $u_i''$ is by definition a $Y$-node. 
Since $u_i''$ is a $Y$-node closer from $r_i$ than $u_i$ and $u_i'$, we have reached a contradiction. 

Now consider the map $\mu:\{T_1,...,T_k\}\rightarrow \nodes(T,Y)$ that sends $T_i$ to $\mu(T_i) = u_i$. 
We claim that the map $\mu$ is an injection, implying in this 
way that $|\{T_1,...,T_k\}|\leq |\nodes(T,Y)|$. Assume for the 
sake of contradiction that for some $i,j$ with $i\neq j$, $\mu(T_i) = \mu(T_j)=u$. Then $u$ is a descendant 
of $r_i$ and a descendant of $r_j$ in $T$. This implies that either $r_i$ is a descendant 
of $r_j$ in $T$, or $r_j$ is a descendant of $r_i$ in $T$. 
Assume that $r_j$ is a descendant of $r_i$ in $T$. Since by assumption $r_i$ and $r_j$ belong to 
distinct connected components in $T\backslash \nodes(T,Y)$, there exists at least one $Y$-node $u'$ in 
in the path from $r_i$ to $r_j$. Therefore, this contradicts the assumption that $u$ is the closest 
descendant of $r_i$ which is a $Y$-node. 
\end{proof}

Let $T_1,...,T_{k}$ be the connected components of $T'$ where $k\leq |\nodes(T,Y)|$. 
For each $i\in \{1,...,k\}$, let $G[T_i]$ be the subgraph of $\graph(\tensornetwork)$
induced by the vertices $\gamma(\leaves(T_i))$. 

\begin{lemma}
\label{lemma:NumberComponentsGraph}
For each $i\in \{1,...,k\}$, the graph $G[T_i]$ has at most 
$2w$ connected components. 
Additionally, there are at most $2w$ edges 
with one endpoint in $G[T_i]$ and another endpoint in $\graph(\tensornetwork)\backslash G[T_i]$. 
\end{lemma}
\begin{proof}
Let $r_i$ be the root of $T_i$ and $u_i$ be the closest descendant of $r_i$ with the 
property that $u_i$ is a $Y$-node. Since by assumption the carving decomposition $(T,\gamma)$  has width $w$, 
we have that $|E(V(r_i))|\leq w$ and $|E(V(u_i))|\leq w$. Suppose
for contradiction that the graph $G[T_i]$ has at least $2w+1$ connected components.
Let $C_{i,1},...,C_{i,c_i}$ be the connected components of $G[T_i]$, where $c_i \geq 2w+1$.
Since the graph $\graph(\tensornetwork)$ is connected, for each $j\in \{1,...,c_i\}$ 
there exists at least one edge with an endpoint in $C_{i,j}$ and another endpoint in 
$V(u_i)\cup (V\backslash V(r_i))$. This implies that $|E(V(u_i))| + |E(V(r_i))| \geq 2w+1$, and therefore 
we have that $|E(V(u_i))|\geq w+1$ or $E(V(r_i))\geq w+1$. But this 
contradicts the assumption that the carving-width of $(T,\gamma)$ is at most $w$. Therefore 
$G(T_i)$ has at most $2w$ connected components.

The proof of the second statement is also by contradiction. Assume that there are at least $2w+1$ edges 
with one endpoint in $G[T_i]$ and other endpoint in $\graph(\tensornetwork)\backslash G[T_i]$. Since 
all vertices in $\graph(\tensornetwork)\backslash G[T_i]$ are mapped to leaves in $\leaves(T)\backslash \leaves(T_i)$, 
we have that $|E(V(u_i))| + |E(V(r_i))|\geq 2w+1$. But then $|E(V(u_i))|\geq w+1$ or $|E(V(r_i))|\geq w+1$. 
This contradicts the assumption that the carving width of $(T,\gamma)$ is $w$.
\end{proof}

Finally, we are in a position to state and prove the main theorem of this section.

\begin{theorem}[Tensor Network Reduction]
\label{theorem:TensorNetworkReduction}
Let $\tensornetwork$ be an algebraic tensor network of carving width $w$ and algebraic degree $d$ 
computing a function $f:\{0,1\}^{\Yvariableset}\rightarrow \{0,1\}$. Let $l\geq |\Yvariableset|$ be the number of
$Y$-tensors in $\tensornetwork$. Then $f$ can be computed by a tensor network 
$\tensornetwork'$ of size at most $4l(w+1)$, rank at most $2w$, and algebraic degree $d$.  
\end{theorem}
\begin{proof}
Let $(T,\gamma)$ be a carving decomposition of $\graph(\tensornetwork)$ of carving width 
at most $w$. Let $\{T_1,...,T_k\}$ be the connected components of the forest $T\backslash \nodes(T,Y)$. 
Let $G[T_i]$ be the subgraph of $\graph(\tensornetwork)$ induced by the vertices 
$\gamma(\leaves(T_i))$. Finally let $C_{i,1},...,C_{i,c_{i}}$ be the connected components of 
$G[T_i]$. We denote by $\tensornetwork[Y]$ the set of $Y$-tensors of $\tensornetwork$. 
Note that if $\atensor$ is not in $\tensornetwork[Y]$ then $\atensor$ has algebraic degree $0$ 
(since no variable in $Y$ occurs in $\atensor$) and labels some vertex of some connected component $C_{i,j}$.
Conversely, each tensor labeling a vertex of a connected component $C_{i,j}$ has algebraic degree $0$. 
For each $i\in \{1,...,k\}$ and each $j\in \{1,...,c_i\}$, let $g_{i,j}$ 
be the tensor obtained by contracting all tensors labeling vertices of the connected 
component $C_{j,i}$. Note that $g_{i,j}$  has algebraic degree $0$ due to the fact that 
$\degree(\tensorContraction(\atensor,\atensor'))\leq \degree(\atensor)+\degree(\atensor')$ for any contractible 
pair of tensors $\atensor,\atensor'$ (Observation \ref{observation:AdditiveDegree}).
Let 
\begin{equation}
\label{equation:ResultingNetwork}
\tensornetwork' = \tensornetwork[Y] \cup \{g_{i,j}\;|\; i\in \{1,...,k\}, j\in \{1,...,c_i\}\}
\end{equation}
be the resulting tensor network. By Lemma \ref{lemma:NumberComponentsTrees}, $k\leq 2\cdot l -1$.
By Lemma \ref{lemma:NumberComponentsGraph}, we have that for each $i\in \{1,...,k\}$, 
$c_i\leq 2w$. Then we have that the number of algebraic tensors in $\tensornetwork'$ 
is at most $l + (2\cdot l -1)\cdot 2w = 4lw - 2w +l < 4l(w+1)$. 
Since algebraic tensors in $\tensornetwork[Y]$ did not get involved into any contraction, both the ranks and the 
algebraic degrees of these algebraic tensors remain unchanged. Therefore, the algebraic degree of 
the network $\tensornetwork'$ is still $d$. 
Now the rank of each new tensor $g_{i,j}$ in $\tensornetwork'$ is equal to the number of 
edges with one endpoint in $G[T_i]$ and another endpoint in $\graph(\tensornetwork)$. By Lemma \ref{lemma:NumberComponentsGraph}
there are at most $2w$ such edges. Therefore, the rank of $g_{i,j}$ is at most $2w$.
\end{proof}

\section{Number of Functions Computable by Tensor Networks of a Given Size, Rank and Algebraic Degree}
\label{subsection:NumberOfFunctionsAlgebraicTensorNetworks}

Let $Y$ be a set of variables. The main result of this section (Lemma \ref{lemma:UpperBoundNumberFunctions}) 
establishes an upper bound on the number of Boolean functions computable by a tensor network over $Y$ of 
size $m$, rank $r$ and algebraic-degree $d$. 

\begin{lemma}
\label{lemma:UpperBoundNumberFunctions}
Let $Y$ be a finite set of variables. 
For each $m,r,d \in \N$ there exists at most $\exp(2^{O(r)}\cdot |Y|^{d+1}  \cdot m\cdot \log m)$ 
Boolean functions $g:\{0,1\}^{Y}\rightarrow \{0,1\}$ which can be computed by some 
algebraic tensor network over $Y$ of size at most $m$, rank at most $r$ and algebraic-degree at most $d$.
\end{lemma}

We will prove Lemma \ref{lemma:UpperBoundNumberFunctions} using 
the {\em connected component counting}  method, an algebraic geometric technique developed 
by Warren in \cite{Warren1968lower}.

\newcommand{\inequality}{\diamond}
\begin{definition}[Sign-Assignment]
\label{definition:SignAssignment}
Let $W$ be a set of variables and let $P = (p_1,p_2,...p_s)$ be a sequence of real polynomials in $\R[W]$.
A $(+,-)$-sign assignment for $P$ is a sequence of inequalities  
$S = (p_1 \inequality_1 0,\;p_2\inequality_2 0,\;...,\;p_s\inequality_s 0)$
where for each $i\in \{1,...,s\}$, $\inequality_i\in \{<,>\}$. 
\end{definition}

\newcommand{\otherassignment}{\beta}

We say that a $(+,-)$-sign assignment $S = (p_1 \inequality_1 0,\;p_2\inequality_2 0,\;...,\;p_s\inequality_s 0)$ 
is {\em consistent} if $S$ is solvable. In other words, $S$ is consistent if there exists an assignment 
$\otherassignment:W\rightarrow \R$ of the variables in $W$ such that for every $i\in \{1,...,s\}$, the 
inequality $p_i(\otherassignment)\diamond_i 0$ is satisfied. 
The following theorem establishes an upper-bound for the number of consistent 
$(+,-)$-sign assignments for a sequence of polynomials $P$ in terms of three parameters: 
the number of variables in $W$, the number of polynomials in $P$, and the maximum degree of a polynomial in $P$. 
Below, $e\approx 2.71$ is the Euler number.  

\begin{theorem}[Warren 1968. Theorem 3 of \cite{Warren1968lower}]
\label{theorem:WarrenTheorem}
Let $P = (p_1,p_2,...,p_s)$ be real polynomials in $\nu$ variables, each of degree 
at most $D\geq 1$. If $s\geq \nu$, then the number of consistent $(+,-)$-sign 
assignments for $P$ is at most $\left(\frac{4\cdot e\cdot D\cdot s}{\nu}\right)^{\nu}$.
\end{theorem}

\newcommand{\forget}{\mathit{forget}}
\newcommand{\typenetwork}{\mathit{type}}

Let $\tensornetwork$ be an algebraic tensor network and let $\graph(\tensornetwork) = (V,E,\tensorlabeling,\indexlabeling)$
be the graph associated with $\tensornetwork$. The type of $\tensornetwork$ is defined as 
$\typenetwork(\tensornetwork) = (V,E)$. In other words, the type of $\tensornetwork$ is the 
{\em unlabeled} graph obtained from $\graph(\tensornetwork)$ by forgetting vertex-labels and edge-labels. 

\begin{proposition}
\label{proposition:NumberOfTypes}
There are at most $m^{r\cdot m}$ types of tensor networks of rank $r$ containing $m$
tensors. 
\end{proposition}
\begin{proof}
Let $\tensornetwork$ be a tensor of rank $r$ containing $m$ tensors. Then 
$\typenetwork(\tensornetwork)$ is a graph with at most $m$ vertices, and 
degree at most $r$. For each vertex $v$ in such a graph, there are at 
most $m^{r}$ ways of connecting $v$ to other $r$ vertices. Therefore,
there are at most $(m^{r})^{m}= m^{r\cdot m}$ such graphs. 
\end{proof}

Let $Y$ be a set of variables. We denote by $\monomials(Y,d)$ the set of 
monomials in $Y$ of degree at most $d$. Note that $|\monomials(Y,d)|\leq |Y|^{d}$.
Now let $G$ be a fixed type of algebraic tensor network of rank $r$ and size $m$. 
We will establish an upper bound on the number of functions computable by
tensor networks over $Y$ of algebraic-degree $d$ and type $G$.
Let $\tensornetwork = [g_1,...,g_m]$ be such a tensor network.
Since $\tensornetwork$ has algebraic degree $d$, 
each entry of each algebraic tensor $\atensor_j$ in $\tensornetwork$ is a complex polynomial 
$p = \sum_{M\in \monomials(Y,d)} (\realpart_M + \complexpart_M\cdot i)M$ in $Y$ of degree at most $d$, 
where $\realpart_M$ and $\complexpart_M$ are real numbers. 
Therefore, each such  polynomial can be specified by at most $2\cdot |Y|^{d}$ real numbers.
Since  $\atensor$ has rank at most $r$, $\atensor_j$ has at most $4^{r}$ entries. Finally,
$\tensornetwork$ has $m$ tensors. Therefore, if we let $\mu = 4^{r}\cdot m\cdot |Y|^d$, 
the whole tensor network $\tensornetwork$ can be specified by a sequence of $2\cdot \mu$ 
real numbers $\realpart_1,...,\realpart_{\mu},\complexpart_1,...,\complexpart_{\mu}$. We let 
$\tensornetwork[\realpart_1,...,\realpart_{\mu},\complexpart_1,...,\complexpart_{\mu}]$ 
be the algebraic tensor network over $Y$ of rank at most $r$, size $m$ and algebraic-degree at most $d$
specified by this sequence.

Now, regard $\realpartvar_1,...,\realpartvar_{\mu},\complexpartvar_1,...,\complexpartvar_{\mu}$ as real 
variables. Then each entry of each tensor in the network 
$\tensornetwork[\realpartvar_1,...,\realpartvar_{\mu},\complexpartvar_1,...,\complexpartvar_{\mu}]$
is a complex polynomial $p$ of degree at most $d+1$ in the variables 
$Y \cup \{\realpartvar_1,...,\realpartvar_{\mu},\complexpartvar_1,...,\complexpartvar_{\mu}\}$.
Additionally, each term of $p$ has a single occurrence of a variable in 
$\{\realpartvar_1,...,\realpartvar_{\mu},\complexpartvar_1,...,\complexpartvar_{\mu}\}$.
Let $\boolassignment:Y\rightarrow \{0,1\}$ be a Boolean assignment for the variables $Y$, and 
let $\tensornetwork[\realpartvar_1,...,\realpartvar_{\mu},\complexpartvar_1,...,\complexpartvar_{\mu}](\alpha)$
be the algebraic tensor network obtained by substituting the value $\boolassignment(x)$ for each 
variable $x\in Y$ occurring in 
$\tensornetwork[\realpartvar_1,...,\realpartvar_{\mu},\complexpartvar_1,...,\complexpartvar_{\mu}]$. Then 
$\tensornetwork[\realpartvar_1,...,\realpartvar_{\mu},\complexpartvar_1,...,\complexpartvar_{\mu}](\alpha)$
is an algebraic tensor network of over the real variables 
$\{\realpartvar_1,...,\realpartvar_{\mu},\complexpartvar_1,...,\complexpartvar_{\mu}\}$  
whose algebraic degree is at most $1$. Therefore, the total degree of this network is at most $m$, 
and by Proposition \ref{proposition:PolynomialFromAlgebraicTensorNetwork}, the polynomial 
\begin{equation}
\label{equation:PolynomialFromNetwork}
\polynomial_{\alpha}(\realpartvar_1,...,\realpartvar_{\mu},\complexpartvar_1,...,\complexpartvar_{\mu}) = 
\finalvalue_{\tensornetwork[\realpartvar_1,...,\realpartvar_{\mu},\complexpartvar_1,...,\complexpartvar_{\mu}](\alpha)}^{2}
\end{equation}
is a real polynomial in $\R[\{\realpartvar_1,...,\realpartvar_{\mu},\complexpartvar_1,...,\complexpartvar_{\mu}\}]$ of 
degree at most $2\cdot m$.

Let $h:\{0,1\}^{Y}\rightarrow \{0,1\}$ be a Boolean function on variables $Y$. 
For each assignment $\alpha\in \{0,1\}^{Y}$, let $\diamond_{\alpha}$ be the {\em greater-than} symbol $>$ if 
$h(\alpha)=1$, and the {\em less-than} symbol $<$ if $h(\alpha)=0$. 
Consider the following system of $2^{|Y|}$ polynomials, indexed by Boolean assignments $\alpha\in \{0,1\}^{Y}$.

\begin{equation}
\label{equation:SystemPolynomials}
\polynomial_{\alpha}(\realpartvar_1,...,\realpartvar_{\mu},\complexpartvar_1,...,\complexpartvar_{\mu}) - 1/4
\diamond_{\alpha} 0, \hspace{1cm} \alpha\in \{0,1\}^Y
\end{equation}

Assume that $h:\{0,1\}^{Y}\rightarrow \{0,1\}$ is computable by an algebraic tensor network of size $m$, rank $r$, 
algebraic degree $d$, and type $G$. Then for some real numbers 
$\realpart_1^{h},...,\realpart_{\mu}^{h},\complexpart_1^{h},...,\complexpart_{\mu}^{h}$
the algebraic tensor network $\tensornetwork_{h} = \tensornetwork[\realpart_1^{h},...,\realpart_{\mu}^{h},\complexpart_1^{h},...,\complexpart_{\mu}^{h}]$
computes $h$. In other words, for each Boolean assignment $\alpha:\{0,1\}^{Y}\rightarrow \{0,1\}$, we have that 
$\finalvalue_{\tensornetwork_h}(\alpha)$ is greater than $1/2$ if $h(\alpha)=1$, and less than $1/2$ if $h(\alpha)=0$.
This implies that $p(\realpart_1^{h},...,\realpart_{\mu}^{h},\complexpart_1^{h},...,\complexpart_{\mu}^{h})$ is greater 
than $1/4$ if $h(\alpha)=1$, and less than $1/4$ if $h(\alpha)=0$. 
Therefore the sequence $\realpart_1^{h},...,\realpart_{\mu}^{h},\complexpart_1^{h},...,\complexpart_{\mu}^{h}$ satisfies 
all inequalities of the system given in Equation \ref{equation:SystemPolynomials}. 

The discussion above shows that the number of Boolean functions computable by an algebraic tensor network over $Y$ of size at most 
$m$, rank at most $r$, algebraic degree at most $d$, and type $G$ is upper bounded by the number of consistent sign assignments
for the system of inequalities of Equation \ref{equation:SystemPolynomials}. Therefore we can use Theorem \ref{theorem:WarrenTheorem}
to estimate this number. 
By setting $s=2^{|Y|}$, $\nu = 2\mu =2\cdot 4^{r}\cdot m\cdot |Y|^d$, and $D = 2m$  in Theorem \ref{theorem:WarrenTheorem}
we have that the number of consistent assignments for the system of polynomials in Equation \ref{equation:SystemPolynomials}
is at most

$$
\left(\frac{4\cdot e\cdot (2m)\cdot 2^{|Y|}}{2\cdot 4^r\cdot m\cdot |Y|^{d+1}}\right)^{2\cdot 4^r\cdot m\cdot |Y|^{d}}
\leq \exp(2^{O(r)}\cdot |Y|^{d+1} \cdot m). 
$$

Therefore, there are at most $\exp(2^{O(r)}\cdot |Y|^{d+1} \cdot m)$ functions computable by some tensor network over $Y$ 
of algebraic degree at most $d$, with type $G$. Since, by Proposition \ref{proposition:NumberOfTypes},
there are at most $m^{r\cdot m} \leq \exp(O(r\cdot m\cdot \log m))$ types of network of rank $r$ and size $m$, 
we have that the total number of 
functions computable by an algebraic tensor network over $Y$ of algebraic-degree $d$, rank $r$ and size $m$ is upper bounded by 
$$\exp(2^{O(r)}\cdot  |Y|^{d+1}  \cdot m + O(r\cdot m\cdot \log m))\leq \exp(2^{O(r)}\cdot  |Y|^{d+1} \cdot m\cdot \log m).$$ 
This proves Lemma \ref{lemma:UpperBoundNumberFunctions}. $\square$

\section{Upper Bounding the Number of Subfunctions of a Function} 
\label{section:MainTechnical}

Let $X=\{x_1,...,x_n\}$ be a set of variables, 
$\afunction:\{0,1\}^X\rightarrow \{0,1\}$ be a Boolean function on $X$, and $Y\subseteq X$
be a subset of variables of $X$.  We denote by $N_{\afunction}(Y)$ the number of distinct functions 
obtained from $f$ by initializing all variables in $X\backslash Y$ with values in $\{0,1\}$. 
Now assume that $f$ is computed by an algebraic tensor network $\tensornetwork$. 
The next theorem establishes an upper bound for $N_{\afunction}(Y)$ in terms of
number of $Y$-tensors in $\tensornetwork$, and in terms of the 
treewidth, rank and algebraic degree of $\tensornetwork$. 

\begin{theorem}[Main Technical Theorem]
\label{theorem:UpperBoundFunctionsTreewidth}
Let $f:\{0,1\}^{X} \rightarrow \{0,1\}$ be a function computable by an algebraic tensor network $\tensornetwork$ 
of treewidth $t$, rank $k$, and algebraic-degree $d$. 
Let $Y\subseteq X$, and $l$ be the number of $Y$-tensors in $\tensornetwork$. 
Then $N_{\afunction}(Y)$ is at most
$\exp\left(2^{O(r\cdot t)}\cdot |Y|^{d+1}\cdot l \cdot \log l \right)$. 
\end{theorem}
\begin{proof}
Let $f:\{0,1\}^{X}\rightarrow \{0,1\}$ be a function computable by an algebraic tensor 
network $\tensornetwork$ over $X$ of rank $r$ and algebraic-degree $d$. Since 
$\graph(\tensornetwork)$ has treewidth $t$ and maximum (vertex) degree $r$, Lemma \ref{lemma:CarvingVsTreewidth}
implies that the carving width of $\graph(\tensornetwork)$ is at most $w = O(r\cdot t)$. 

Let $Y\subseteq X$, and $l$ be the number of $Y$-tensors in $\tensornetwork$. 
Let $\beta:\{0,1\}^{X\backslash Y}\rightarrow\{0,1\}$ be an assignment 
of the variables in $X\backslash Y$, and let $\tensornetwork(\beta)$ be the algebraic tensor network over $Y$, 
obtained by initializing the variables in $X\backslash Y$ according 
to the assignment $\beta$. Then $\tensornetwork(\beta)$ computes the function $g:\{0,1\}^{Y}\rightarrow \{0,1\}$
which is obtained from $f$ by restricting the variables in $X\backslash Y$ according to $\beta$.

By Theorem \ref{theorem:TensorNetworkReduction}, the function $g$ can be computed by an algebraic tensor network 
$\tensornetwork'$ over $Y$ of algebraic degree $d$, rank $r' = O(r\cdot t)$, and size 
$m = O(r\cdot t\cdot l)$. Therefore, by Lemma \ref{lemma:UpperBoundNumberFunctions}
we have that there exist at most 
$$\exp\left(2^{O(r\cdot t)}\cdot |Y|^{d+1}\cdot O(r\cdot t\cdot l \cdot \log (r\cdot t\cdot l))\right) = 
\exp\left(2^{O(r\cdot t)}\cdot |Y|^{d+1}\cdot l \cdot \log l \right) 
$$
Boolean functions $g:\{0,1\}^{Y}\rightarrow \{0,1\}$ which can be obtained from $f$ by initializing the variables 
in $X\backslash Y$ with elements from $\{0,1\}$. 
\end{proof}

\section{Quadratic Lower Bounds For Algebraic Networks and Quantum Circuits of Constant Treewidth}
\label{section:Nechiporuk}

Let $X = \{x_1,...,x_{n}\}$ be a set of $n=2k\log k$ distinct variables partitioned into $k$ blocks 
$Y_1,Y_2, ..., Y_k$, where each block $Y_i$ has $2\log k$ variables. The {\em element distinctness} 
function $\edfunction_n:\{0,1\}^{X}\rightarrow \{0,1\}$ is defined as follows for 
each assignment $s_1,s_2, ... ,s_k$ of the blocks $Y_1, Y_2, ... ,Y_k$ respectively.

\begin{equation}
\label{equation:ElementDistinctness}
\edfunction_n(s_1,s_2,...,s_k) = \left\{ \begin{array}{lll} 
1 & & \mbox{if $s_i\neq s_j$ for $i\neq j$}, \\
0 & & \mbox{otherwise.}
\end{array}\right.
\end{equation} 

The following lemma states that the element distinctness function defined in Equation \ref{equation:ElementDistinctness} 
has many sub-functions.

\begin{lemma}[\cite{Jukna2012}, Section 6.5]
\label{lemma:ElementDistinctness}
Let $\edfunction_n:\{0,1\}^X\rightarrow \{0,1\}$ be the element distinctness function defined in 
Equation \ref{equation:ElementDistinctness}, where $|X|=n$ and 
$X=Y_1\;\dot\cup\; Y_2\;\dot\cup ... \;\dot\cup\; Y_k$ with $|Y_i|=2\log k$. 
Then for each $i\in \{1,...,k\}$, $N_{\edfunction_n}(Y_i) \geq 2^{\Omega(n)}$. 
\end{lemma}

The following theorem follows as a combination of Theorem \ref{theorem:UpperBoundFunctionsTreewidth} and 
Lemma \ref{lemma:ElementDistinctness}.

\begin{theorem}
\label{theorem:MainTheoremAlgebraicNetwork}
Let $X$ be a set with $n$ Boolean variables, and let $\delta_n:\{0,1\}^{X}\rightarrow \{0,1\}$ be 
the $n$-bit element distinctness function. Let $\tensornetwork$ be a tensor network of
treewidth $t$, rank $r$ and algebraic degree $d$ computing $\delta_n$. Then $\tensornetwork$ has size 
$$\Omega\left(\frac{n^{2}}{2^{O(r\cdot t)}\cdot (\log n)^{d+3}}\right).$$
\end{theorem}
\begin{proof}
For each $i\in \{1,...,k\}$ let $l_i$ be the number of $Y_i$-nodes in $\tensornetwork$ where 
$Y_i$ is the $i$-th block of variables. If $l_i\geq n^2$, then the theorem is true and there 
is nothing to be proved. Therefore, assume that $l_i<n^2$, and hence that $\log l_i< 2\log n$. 
For each $i\in \{1,...,k\}$, by plugging 
$l_i$ and $|Y_i| = 2\log n$ in Theorem \ref{theorem:UpperBoundFunctionsTreewidth}, we have that
\begin{equation}
\label{equation:NumberFunctions}
N_{\delta_n}(Y_i)\leq \exp\left(2^{O(r\cdot t)}\cdot (\log n)^{d+1} \cdot l_i\cdot \log l_i\right)\leq 
\exp\left(2^{O(r\cdot t)}\cdot (\log n)^{d+2} \cdot l_i \right).
\end{equation}

Now, by Lemma \ref{lemma:ElementDistinctness}, we have that $N_{\delta_n}(Y_i)\geq 2^{\Omega(n)}$, and 
therefore, 

\begin{equation}
\label{equation:Inequalities}
\exp\left(2^{O(r\cdot t)}\cdot (\log n)^{d+2} \cdot l_i  \right) \geq N_{\delta_n}(Y_i) \geq 2^{\Omega(n)}.
\end{equation}

Equation \ref{equation:Inequalities} implies that 
$$l_i \geq \Omega\left(\frac{n}{2^{O(r\cdot t)}\cdot (\log n)^{d+2}}\right).$$ 
Since there are $k=\Omega(\frac{n}{\log n})$ blocks of variables $Y_i$, 
we have that the total number of tensors in $\tensornetwork$, which is greater than $\sum_i l_i$, 
is at least 
$$\Omega\left(\frac{n^2}{2^{O(r\cdot t)}\cdot (\log n)^{d+3}}\right).$$ 
\end{proof}

Finally, our main theorem follows as a corollary of Theorem \ref{theorem:MainTheoremAlgebraicNetwork}.

\begin{theorem}[Main Theorem]
\label{theorem:MainTheoremQuantumCircuits}
Let $X$ be a set with $n$ Boolean variables, and let $\delta_n:\{0,1\}^{X}\rightarrow \{0,1\}$ be 
the $n$-bit element distinctness function. Let $C$ be a quantum circuit over $X$ computing 
$\edfunction_{n}$. If $C$ has treewidth $t$ and all gates in $C$ act on at most 
$r$ qubits, then $C$ has at least 
$\Omega\left(\frac{n^{2}}{2^{O(r\cdot t)}\cdot (\log n)^{4}}\right)$ gates. 
\end{theorem}
\begin{proof}
Let $\tensornetwork_C$ be the algebraic tensor network associated with $C$. Then 
$\tensornetwork_C$ has algebraic degree $1$, treewidth $t$, and rank at most $2\cdot r$.
By Theorem \ref{theorem:MainTheoremAlgebraicNetwork}, $\tensornetwork_C$ must 
have at least $\Omega\left(\frac{n^2}{2^{O(r\cdot t)}\cdot (\log n)^{4}}\right)$ tensors, 
and therefore $C$ must have at least this number of gates. 
\end{proof}

\section{Final Comments and Open Problems}

In this work we have shown that any quantum circuit of treewidth at most $t$, 
build up from $r$-qubit gates, requires at least $\Omega(n^2/2^{O(r\cdot t)}\log^{4} n)$ gates 
to compute the element distinctness function $\edfunction_n:\{0,1\}^{n}\rightarrow \{0,1\}$ 
(Theorem \ref{theorem:MainTheoremQuantumCircuits}). This lower bound is robust for three reasons. 
First, it does not assume that the quantum gates belong to any particular finite basis. The only 
requirement is that these gates act on at most $r$ qubits. Second, we do not assume any 
upper bound on the number of bits necessary to represent each entry of such a gate. 
Third, we consider that a function $f:\{0,1\}^{X}\rightarrow \{0,1\}$ is computed by a quantum 
circuit $C$ if the acceptance probability of $C$ on input $\alpha\in \{0,1\}^{X}$ is greater 
than $1/2$ whenever $f(\alpha)=1$, and less than $1/2$ whenever $f(\alpha)=0$. Thus we 
assume no gap between the acceptance and rejection probabilities for a given input $\alpha$.

There are many interesting open problems concerning circuits of constant treewidth. For instance, 
can quantum circuits of treewidth $t$ be polynomially simulated by quantum (or classical) circuits of 
treewidth $t-1$? Can quantum circuits of treewidth $t$ be polynomially simulated by quantum formulas (i.e. quantum circuits 
of treewidth $1$)? Also, we should mention the longstanding open problem of determining whether quantum 
formulas can be polynomially simulated by classical formulas \cite{RoychowdhuryVatan2001quantum}. Progress
towards this question has only been made in the read-once setting. More precisely, it 
has been shown that read-once quantum formulas can be polynomially simulated by classical formulas of same 
size built from Toffoli and NOT gates \cite{ConsentinoKothariPaetznick2013}. Nevertheless this 
simulation breaks down if the read-once condition is removed \cite{ConsentinoKothariPaetznick2013}. 
It would be interesting to determine whether a similar result can be achieved for read-once quantum circuits of constant treewidth. 
Can read-once quantum circuits of treewidth $t$ be polynomially simulated by read-once classical circuits 
of treewidth $t$?

\subsection{Acknowledgements}

The author thanks Christian Komusiewicz for valuable comments and suggestions. 
The author acknowledges support from the Bergen Research Foundation. Part of this work was 
done while the author was at the Czech Academy of Sciences, supported by the 
European Research Council (grant number 339691).

\bibliographystyle{abbrv}
\bibliography{quantumNechiporuk}

\appendix

\section{Proof of Proposition \ref{proposition:ConversionQuantumCircuitsAlgebraicTensorNetworks}}
\label{appendix:ConversionQuantumCircuitsAlgebraicTensorNetworks}

In this section we show that any quantum circuit $C$ with $m$ gates, treewidth $t$, build from $r$-qubit gates, 
can be simulated by an algebraic tensor network $\tensornetwork_{C}$ with $m$ algebraic tensors, 
treewidth $t$, rank $2r$, and algebraic degree $1$. The construction of $\tensornetwork_{C}$ from 
$C$ is based on a construction given in \cite{MarkovShi2008} which converts quantum circuits 
in which all inputs are initialized to tensor networks (i.e. algebraic tensor networks of degree $0$).
Below, we modify this construction to take into consideration 
input vertices that are are labeled with variables. 

Let $C = (V,E,\vertexlabelingfunction,\edgelabelingfunction)$ be a quantum circuit over a set of variables 
$\variableset$. The tensor network $\tensornetwork_{C}$ is obtained by creating a tensor $g_v$ for each 
vertex $v\in V$. 

Let $v$ be an internal vertex of $C$ whose incoming edges are labeled with numbers $\{i_1,...,i_k\}$ and 
outgoing edges are labeled with numbers $\{j_1,...,j_k\}$. Let $v$ be labeled with a unitary matrix 
$U\in \C^{2^{k}\times 2^{k}}$. Then the tensor $g_v$ has index set $\{i_1,...,i_k,j_1,...,j_k\}$, and
the value of $g_v$ on each entry $\sigma_{i_1},...,\sigma_{i_k},\sigma_{j_1},...,\sigma_{j_k}\in \Pi^{2k}$ is defined 
as follows. 

\begin{equation}
\atensor_v(\sigma_{i_1},...,\sigma_{i_k},\sigma_{j_1},...,\sigma_{j_l}) = 
\trace\left( [\sigma_{j_1}^{\dagger}\otimes ...\otimes \sigma_{j_l}^{\dagger}]\cdot  U \cdot 
[\sigma_{i_1}\otimes ...\otimes \sigma_{i_k}] \right).
\end{equation}

Let $v$ be an output vertex whose unique incoming edge is labeled with number $j$. Let 
$v$ be labeled with a $1$-qubit measurement element $M$ in $\C^{2\times 2}$. Then the 
tensor $g_v$ has index set $\{j\}$, and the value of $g_v$ on each entry $\sigma_j\in \Pi$ 
is defined as follows. 

\begin{equation}
\atensor_v(\sigma_{j}) = 
\trace\left( \sigma_{j}^{\dagger} \cdot M \right).
\end{equation}

For each variable $x$ we define the following matrix: 
$\ket{x}\bra{x} = \left[\begin{array}{cc}
(1-x) & 0 \\
0 & x \\
\end{array}\right]$. If $v$ is an input vertex of $C$ whose unique outgoing edge is labeled
with number $i$, then we the tensor $\atensor_v$ has index set $i$, and the value of $\atensor_v$
on each entry $\sigma_i\in \Pi$ is defined as follows.

\begin{equation}
\atensor_v(\sigma_{i}) = 
\trace\left( \ket{x}\bra{x}\cdot \sigma_i \right).
\end{equation}

Note that the tensor $\atensor_v$ has algebraic degree $1$. 
On the other hand if such an input vertex $v$ is labeled with a qubit $\ket{b}\in \{\ket{0},\ket{1}\}$, 
then the value of $\atensor_v$ on each entry $\sigma_i\in \Pi$ is defined as. 

\begin{equation}
\atensor_v(\sigma_{i}) = 
\trace\left( \ket{b}\bra{b}\cdot \sigma_i \right).
\end{equation}

We note that if all gates in $C$ act on at most $k$ qubits, then 
the tensor network $\tensornetwork_{C}$ has rank at most $2k$. Additionally, 
the graph $G(\tensornetwork)$ is isomorphic to the graph $\graph(C)$. Therefore, 
if $C$ has treewidth $t$, then $\graph(C)$ has also treewidth $t$. We also note 
tensors associated with input nodes of $C$ labeled with variables have 
algebraic degree $1$. All other tensors have algebraic degree $0$. Therefore,
$\tensornetwork_{C}$ has algebraic degree $1$. $\square$

\end{document}